\documentclass[letterpaper, 10 pt, conference]{ieeeconf}  
\IEEEoverridecommandlockouts                              
\usepackage{float,amsmath,amsfonts,amssymb,amsthm,color,graphicx}
\usepackage{graphics} 
\usepackage{booktabs} 
\usepackage{caption, subcaption}
\usepackage{breqn}
\usepackage{array, mathtools}
\usepackage{pgfplots}
\usepackage{siunitx}
\usepackage{algorithm}
\usepackage[noend]{algpseudocode}
\usepackage{tabulary}
\usepackage{stmaryrd} 
\usepackage{varwidth}

\makeatletter
\newsavebox{\@brx}
\newcommand{\llangle}[1][]{\savebox{\@brx}{\(\m@th{#1\langle}\)}%
  \mathopen{\copy\@brx\kern-0.5\wd\@brx\usebox{\@brx}}}
\newcommand{\rrangle}[1][]{\savebox{\@brx}{\(\m@th{#1\rangle}\)}%
  \mathclose{\copy\@brx\kern-0.5\wd\@brx\usebox{\@brx}}}
\makeatother

\usepackage{tikz}
\usepackage{pgfplots}
\usetikzlibrary{decorations.markings}
\usetikzlibrary{patterns}
\usetikzlibrary{arrows}
\usetikzlibrary{shapes}
\usetikzlibrary{fit}
 \usetikzlibrary{calc}
\usetikzlibrary{decorations.pathreplacing}
\usetikzlibrary{arrows,positioning} 
\usetikzlibrary{pgfplots.groupplots}
\usepackage{soul}
\usepackage{thmtools}
\declaretheorem[style=definition,name=Definition,qed=$\blacksquare$]{definition}
\declaretheorem[style=definition,name=Example,qed=$\blacksquare$]{example}

\declaretheorem[style=definition,name=Problem,qed=$\blacksquare$]{problem}

\declaretheorem[style=definition,name=Remark,qed=$\blacksquare$]{remark}
\declaretheorem[style=plain,name=Proposition]{proposition}
\declaretheorem[style=plain,name=Theorem]{theorem}

\DeclareMathOperator*{\argmin}{arg\, min} 

\DeclareMathOperator*{\argmax}{arg\, max} 

\newcommand{\R}{\mathbb{R}}

\renewcommand{\S}{\mathcal{S}}

\newcommand{\A}{\mathcal{A}}

\title{\LARGE {\bf A Numerical Method to Compute Stability Margins of Switching Linear Systems}}

\author{Corbin Klett, Matthew Abate, Samuel Coogan, and Eric Feron
\thanks{This work was supported by the KAUST baseline budget. 
}
\thanks{C. Klett is with the School of Aerospace Engineering, Georgia Institute of Technology, Atlanta, 30332, USA: 
{\tt\small Corbin@GaTech.edu}.}
\thanks{M. Abate is with the School of Mechanical Engineering and the School of Electrical and Computer Engineering, Georgia Institute of Technology, Atlanta, 30332, USA: {\tt\small Matt.Abate@GaTech.edu}.}
\thanks{S. Coogan is with the School of Electrical and Computer Engineering and the School of Civil and Environmental Engineering, Georgia Institute of Technology, Atlanta, 30332, USA: 
{\tt\small Sam.Coogan@GaTech.edu}.}
\thanks{E. Feron is with the Department of Electrical Engineering, King Abdullah University of Science and Technology, Thuwal, Saudi Arabia: {\tt\small Eric.Feron@Kaust.edu.sa}.}
}

\begin{document}
\def\reals{\mathbb{R}}
\def\Co{{\bf Co}}
\def\diag{{\bf diag}}

\maketitle
\thispagestyle{empty}
\pagestyle{empty}

\begin{abstract}
Stability margins for linear time-varying (LTV) and switched-linear systems are  traditionally computed via quadratic Lyapunov functions, and these functions certify the stability of the system under study. In this work, we show how the more general class of homogeneous polynomial Lyapunov functions is used to compute stability margins with reduced conservatism, and we show how these Lyapunov functions aid in the search for periodic trajectories for marginally stable LTV systems.  Our work is premised on the recent observation that the search for a homogeneous polynomial Lyapunov function for some LTV systems is easily encoded as the search for a quadratic Lyapunov function for a related LTV system, and our main contribution is an intuitive algorithm for generating upper and lower bounds on the system's stability margin. We show also how the worst-case switching scheme---which draws an LTV system \emph{closest} to a periodic orbit---is generated. Three numerical examples are provided to aid the reader and demonstrate the theoretical contributions of the work.
\end{abstract}

\section{Introduction}


Switched-linear  and linear time-varying (LTV) systems are two important and widely studied classes of systems \cite{liberzon2003switching}. They are useful abstractions for studying systems with uncertain parameters, and they can be used to represent mode-switched systems when the switching scheme may be unknown. In all cases, computing stability margins is an important exercise, with practical benefits for safety critical systems \cite{afman2018triple}. In the following,  we study both switched-linear and LTV systems in an equivalent manner and propose methods to compute  upper and lower bounds on a system's stability margin.

We study the stability of LTV systems whose dynamics are influenced by an uncertain system parameter.  We characterize the system's robustness with respect to this parameter by computing an upper and a lower bound such that the system is always stable when the parameter is less than the lower bound and such that a diverging trajectory exists when the parameter is greater than the upper bound. As will be shown, these bounds hint at the existence of a specific parameter value that induces periodic orbits for LTV systems but does not induce diverging trajectories. For this reason, we approach the problem by studying methods for finding periodic trajectories in marginally  LTV systems.

Stability margins are generally impossible to compute exactly for all but linear time-invariant systems,  though much work has been done on estimating the margin through the computation of upper and lower bounds \cite{young1997lower}\cite{feronBoussios}. For time-varying systems, while quadratic stability has been studied in depth since the late 1940s by Lur'e, Yakubovitch, and others~\cite{BEFB:94}, exact approaches to the computation of stability margins are more recent and include~\cite{Rap:90} and~\cite{Rap:93}. However, most computationally tractable analyses are generally conservative and typically give results which may significantly under- or over-approximate the true stability margin. 

In order to gain an understanding of the conservatism of these types of bounding methods, we compute a less-conservative bound by searching for a trajectory which becomes unstable when the system is allowed to operate near its stability margin. Such a trajectory is referred to as a \textit{counter-example} and is a powerful tool for numerically approximating the true stability margin. Generating a counter-example can be challenging. In many cases, simulations will show that a time-varying system operating near or even beyond its stability margin is always stable. Work by Parrilo and Jadbabaie investigates computation of worst-case switching sequences for discrete time systems \cite{parrilo_jsr}. Our work takes inspiration from this approach when computing counter-examples for continuous LTV systems.

In the case where the system matrix remains Hurwitz for all time, the examples studied indicate that instability occurs through a switching sequence which leads the system through a Hopf bifurcation as the uncertain parameter is adjusted. A theoretical result which confirms this observation is not achieved, but the observation nonetheless provides the intuition behind the algorithm for computing a counter-example and therefore an upper bound on the stability margin. In classical linear systems analysis, a system which depends on a single parameter moves from stable to unstable when, as the parameter is adjusted, the poles on the root locus cross through the origin or through the imaginary axis. In the latter case, instability is achieved via a Hopf bifurcation. For the LTV system currently under study, the system matrix can remain Hurwitz so that all poles stay in the left-half plane for all time, but the time-dependent variation of the uncertain parameter can induce a limit cycle. 

Our approach to bounding stability margins is aided by the  recently presented hierarchy of polynomial Lyapunov functions for switched-linear systems \cite{abate2019lyapunov}. In that work, the search for homogeneous polynomial Lyapunov functions for LTV systems is recast as a search for quadratic Lyapunov functions for a related hierarchy of time-varying Lyapunov differential equations; thus, this is an elegant and equivalent alternative to sum-of-squares programming techniques for computing Lyapunov functions \cite{parrilo2000structured}. Intuitive procedures following from these observations have been shown to provide higher-fidelity performance analyses \cite{abate2020performance} \cite{miller2020peak}. The Lyapunov functions returned by these analyses also allow us to improve stability margin analysis by providing valuable information about how to compute counter-examples. A numerical example in this paper shows how a Lyapunov function for an LTV system can be used to produce a trajectory which approaches the bound computed in \cite{abate2020performance}.

This paper is structured as follows. In Section II, we review mathematical preliminaries and the Kronecker product-based formulation of the LTV system hierarchy. Section III defines the system under study and defines the problems addressed in this work. In Section IV, the hierarchy of quadratic Lyapunov functions is used to compute a lower bound on the system stability margin. Sections V and VI use these Lyapunov functions to develop the algorithm for producing a counter-example and an upper bound on the stability margin. Numerical examples are presented in section VII.

\section{Notation and Preliminaries}
We denote by $S_{++}^n \subset \R^{n\times n}$ the set of symmetric positive definite $n\times n$ matrices. We denote by $I_n$ the $n\times n$ identity matrix, and we denote by $0_n \in \R^n$ the zero vector in $\R^n$.  Given $M \in \R^{n\times m}$ and integer $i \geq 1$, we denote by $\otimes^i M \in \R^{n^i \times m^i}$ the $i^{\rm th}$-Kronecker power of $M$, as defined recursively by 
\begin{equation}
\begin{split}
    \otimes^1 M &:= M \\
    \otimes^i M &:= M \otimes (\otimes^{i-1} M) \qquad i \geq 2.
\end{split}    
\end{equation}

We study, in the following, switched linear systems, as in
\begin{equation}\label{LTV_sys}
    \begin{split}
    \dot{x} & = A(t)x,
    \end{split}
\end{equation}
with state $x\in \R^n$ and where $A(t) \in \S \subseteq \R^{n\times n}$ evolves for all time inside a finite set of switched modes
\begin{equation}\label{convexhull}
    \S = \{A_1,\, \cdots,\, A_k\}.
\end{equation}
We are chiefly interested in studying the stability properties of \eqref{LTV_sys}.  

\begin{definition}[Stability]
The system \eqref{LTV_sys} is \emph{stable} if for all initial conditions $x(0) \in \R^n$, we have that $\lim_{t \rightarrow \infty} x(t)= 0_n$ for all $A(t)$ satisfying \eqref{convexhull}.
\end{definition}

The system \eqref{LTV_sys} is \emph{unstable} when \eqref{LTV_sys} is not stable.

\begin{definition}[Quadratic Stability]\label{def1}
For $P\in S^n_{++}$, the function $V(x) = x^TPx$ is a quadratic Lyapunov function for \eqref{LTV_sys} if 
\begin{equation}
    A^TP + PA \prec 0 \quad\text{for all} \quad A\in \S.
\end{equation}
\end{definition}

When \eqref{LTV_sys} is \emph{linear time-invariant}---equivalently, when $k = 1$---the system \eqref{LTV_sys} is stable if and only if there exists a quadratic Lyapunov function for \eqref{LTV_sys}. This is, however, not true in the general setting of \eqref{LTV_sys}; indeed, there exists stable switched systems for which no quadratic Lyapunov function exists \cite{liberzon2003switching}.  Nonetheless, the existence of such a quadratic Lyapunov function for \eqref{LTV_sys} guarantees the stability of the system.

It was recently shown in \cite{abate2019lyapunov} how a related hierarchy of switched linear systems is used to assess the stability of \eqref{LTV_sys}. Consider the following infinite hierarchy of switched linear systems:
\begin{equation}\label{hierarchy}
\begin{array}{l}
H_1: \left\{
\begin{array}{rcl} 
\dot{\xi}_1 &=& \A_1(t)\xi_1\\
\A_1(t)&\in&\S_1\\
S_1 &=& \{\A_1^1,\, \cdots,\, \A_k^1\} \\
\A_j^1&=& A_j 
\end{array}\right.\\[10pt] \\
H_i:  \left\{
\begin{array}{rcl} 
\dot{\xi}_i &=& \A_i(t)\xi_i\\
\A_i(t) &\in& \S_i \\
\S_i &=&  \{\A_1^i,\, \cdots,\, \A_k^i\}\\
\A^i_j &=& I_n \otimes \A_j^{i-1} + A_j\otimes I_{n^{i-1}}
\end{array}\right.
\end{array}
\end{equation}
where $\xi_i \in \R^{n^i}$ is the state of the $i^{\rm th}$-level system $H_i$ and $i \geq 1$.  Each system $H_i$ in the hierarchy is switched linear, as in \eqref{LTV_sys}.

The hierarchy \eqref{hierarchy} is best understood by looking at the first and second level systems $H_1$ and $H_2$.  The system $H_1$ is equivalent to the switched linear system \eqref{LTV_sys}, and the system $H_2$ is nothing but the vectorized version of the Lyapunov differential equation
\begin{equation}\label{LDE}
    \dot{X}(t) = A(t)X(t)+X(t)A(t)^T
\end{equation}
where $A(t)$ maintains its definition from \eqref{LTV_sys} and where the state of \eqref{LDE} is a matrix $X\in \R^{n\times n}$. Moreover, it is shown in \cite{abate2019lyapunov} that if $x(t)$ is a solution to \eqref{LTV_sys}, then $\xi_i(t) = (\otimes^i x(t))$ is a solution to $H_i$. It follows, therefore, that if the $i^{\rm th}$-level $H_i$ is stable, then the system \eqref{LTV_sys} is stable as well.  

\begin{proposition}[ \hspace{1sp}\label{prop1text}  \cite{abate2019lyapunov} \hspace{1sp}]
If the $i^{\rm th}$-level system $H_i$ is quadratically stable, i.e. if there exists a $P \in S^{n^i}_{++}$ such that 
\begin{equation}\label{prop1}
    \A^TP + P\A \prec 0 \quad\text{for all} \quad \A\in \S_i,
\end{equation}
then the system \eqref{LTV_sys} is stable and 
\begin{equation}\label{lyap}
    V(x) = (\otimes^i x)^TP(\otimes^i x)
\end{equation}
is a homogeneous polynomial Lyapunov function for \eqref{LTV_sys} of order $2i$.
\end{proposition}

Note that the system \eqref{LTV_sys} can be stable even when $H_i$ is not quadratically stable for a given $i\geq 1$.  Nonetheless, the stability of \eqref{LTV_sys} implies that there must exist an $i\geq 1$ such that $H_i$ is quadratically stable \cite{abate2019lyapunov}.

\section{Problem Statement}
In the following, we consider linear time-varying (LTV) systems
\begin{equation}\label{uncertain}
\begin{split}
    \dot{x} &= A(t)x \\
    A(t) &= A + \Delta(t) A_0
\end{split}
\end{equation}
for $A,\, A_0 \in \R^{n\times n}$ and time-varying scalar $\Delta(t) \geq 0$.  We refer to $A$ as the nominal dynamics, and the term $\Delta(t)A_0$ describes a time-varying perturbation. We assume that $A$ is Hurwitz so that \eqref{uncertain} is stable when $\Delta(t) \equiv 0$.

\begin{remark}\label{remark1}
The switched linear system \eqref{LTV_sys} is a generalization of \eqref{uncertain}; in particular, if $\S$ in \eqref{convexhull} is taken as
\begin{equation}\label{convexhull2}
    \S = \{A,\, A + \delta A_0\},
\end{equation}
then the switched system \eqref{LTV_sys} is stable if and only if the linear time varying system \eqref{uncertain} is stable for all $\Delta(t)$ that satisfy $\Delta(t) \in [0,\, \delta]$ for all $t \geq 0$.
\end{remark}

\begin{definition} \label{stablewrt}
Given $A_0,\, A \in \R^{n\times n}$ and $\delta \in \R$, the system \eqref{uncertain} is stable with respect to $\delta$ if \eqref{uncertain} is stable for all $\Delta(t)$ such that $\Delta(t) \in [0,\, \delta]$ for all $t\geq 0$. The \textit{stability margin} for \eqref{uncertain} is the unique $\hat{\delta}\geq 0$ such that \eqref{uncertain} is stable with respect to all $\delta \in [0,\, \hat{\delta})$ and such that \eqref{uncertain} is not stable with respect to $\hat{\delta}$.
\end{definition}

The system \eqref{uncertain} is guaranteed to be stable with respect to all $\delta \in [0,\, \hat{\delta})$, and for this reason, we begin by studying methods for under-approximating the stability margin for \eqref{uncertain}.

\begin{problem}\label{prob1}
Under-approximate the stability margin $\hat{\delta}$ for \eqref{uncertain} by finding the largest possible $\delta$ for which a Lyapunov function can be computed which certifies that \eqref{uncertain} is stable with respsect to $\delta$. We denote this lower bound on $\hat{\delta}$ as $\underline{\delta}$.
\end{problem}

Numerous methods exist for computing $\underline{\delta}$ \cite{BEFB:94}, but these methods generally provide conservative estimates of the stability margin $\hat{\delta}$. In Section \ref{sec1}, we present an intuitive procedure for solving Problem \ref{prob1} whereby the stability margin is under-approximated using a homogeneous polynomial Lyapunov function for \eqref{uncertain}.  These Lyapunov functions are computed using the hierarchy of switched systems \eqref{hierarchy}, and we show how conservatism is mitigated in our approach when high-order Lyapunov functions are considered.

Next, consider the case that $A+\delta A_0$ remains Hurwitz and assume that, for $\delta = \hat{\delta} + \varepsilon$, there exists an $\varepsilon > 0$ and a switching function $\Delta(t)$ which produces a diverging trajectory. Based on the examples studied, we observe that, when $\varepsilon = 0$, there exists a $\Delta(t)$ that can induce a periodic trajectory but not a diverging trajectory. Furthermore, for any $\varepsilon > 0$ there exists a $\Delta(t)$ which can produce a diverging trajectory. This observation has not been proven; however, it inspires the development of a useful algorithm for bounding the stability margin from above and finding a periodic trajectory.

\begin{definition}
A switching signal $\Delta(t) \in [0,\, \delta]$ which can cause a trajectory of \eqref{uncertain} to become periodic when $\delta=\hat{\delta}$ is called a \textit{worst-case switching function} and is denoted as $\Delta_w(t;\delta)$. The system matrix $A(t)$ produced by the worst-case switching function is denoted by $A_w(t;\delta) = A + \Delta_w(t;\delta) A_0$.
\end{definition}

\begin{remark}
It is not proven that $\Delta_w(t;\delta)$ always exists, but in practice it has been approximated for every example LTV system studied. It is also not necessarily unique, but that is not important since any $\Delta_w(t;\delta)$ will produce a periodic trajectory when $\delta=\hat{\delta}$, and it is the approximation of this $\hat{\delta}$ which is of primary interest.
\end{remark}

Even for $\delta < \hat{\delta}$, a worst-case switching function $\Delta_w(t;\delta)$ is very useful for system analysis as it can be used to slow a trajectory's path to the origin. See Example \ref{ex3} in Section VII.

While $\Delta_w(t;\delta)$ can not be computed exactly, it can be closely approximated using a Lyapunov function for \eqref{uncertain}, and a higher-order Lyapunov function will produce a better approximation. We observe that the transition from stability to instability occurs through a switching sequence which results in a periodic orbit about the surface of a level set of a Lyapunov function.

\begin{problem}
Given $A,\, A_0$, $\delta$ and a Lyapunov function parameterized by $P$ from the solution to Problem \ref{prob1}, approximate $\Delta_w(t;\delta)$ and $A_w(t;\delta)$. These approximations are denoted as $\Delta_w(t;\delta,P)$ and $A_w(t;\delta,P)$. Often, the order of the Lyapunov function used to generate the switching function is included as a subscript, as in $\Delta_w(t;\delta,P_{2i}).$
\end{problem}

Once an approximate worst-case switching function is found, the smallest value of $\delta$ for which $\Delta_w(t;\delta,P)$ produces a periodic trajectory can be computed.

\begin{problem}\label{prob3}
Find the smallest value of $\delta$ such that the system \eqref{uncertain} with $A_w(t;\delta,P)$ produces a periodic trajectory. In this case, $\delta$ approximates the lowest upper bound on the stability margin $\hat{\delta}$. This upper bound is denoted as $\overline{\delta}$.
\end{problem}

\section{Under-Approximating the Stability Margin}\label{sec1}

We begin by addressing Problem \ref{prob1}. The system \eqref{uncertain} is stable if and only if there exists a homogeneous sum-of-squares polynomial Lyapunov function certifying its stability \cite{ahmadi2017sum} \cite{mason2006common}.  Thus, we present an iterative approach for under-approximating $\hat{\delta}$, and this approach relies on the hierarchy of switched systems \eqref{hierarchy}.

\begin{proposition}\label{thing5}
Assume $i \geq 1$ and $\delta \geq 0$, and define $\S = \{A,\,A + \delta A_0\}$. If there exists a  $P \in S_{++}^{n^i}$ satisfying the constraint \eqref{prop1}, then $\delta < \hat{\delta}$ and \eqref{uncertain} is stable with respect to $\delta$.  Moreover, $\underline{\delta} = \delta$ solves Problem \ref{prob1}.
\end{proposition}

Proposition \ref{thing5} is a direct result of Proposition \ref{prop1text} and the discussion provided in Remark \ref{remark1}.
Further, Proposition \ref{thing5} shows how the uncertain system \eqref{uncertain} is certified stable with respect to a candidate parameter $\delta$; one can certify that $\delta$ solves Problem \ref{prob1} by computing a quadratic Lyapunov function for the $i^{\rm th}$-level system $H_i$.  Certifying a candidate parameter $\delta$ in this way may lead to conservative estimates of the stability margin $\hat{\delta}$. However, we show next how the homogeneous polynomial Lyapunov function, which is derived in the certification of $\delta$, can be used to compute alternate parameters and, in this way, reduce conservatism.

\begin{proposition}\label{corbin2}
    For a given $i \geq 1$ and $\delta \geq 0$, assume that there exists a $P \in S_{++}^{n^i}$ satisfying the constraint \eqref{prop1} for $\S = \{A,\,A + \delta A_0\}$.  Define by 
    \begin{equation}\label{deltaP}
        \delta_{P} := \argmax\limits_{\delta^{*} \geq \delta} \delta^* 
    \end{equation}
    \begin{equation}\label{const5}
        \text{s.t.} \qquad (\A+\delta^{*}\A_0)^TP + P(\A+\delta^{*}\A_0) \preceq 0
    \end{equation}
where
\begin{equation}
\begin{split}
    \mathcal{A} &:= \sum_{ j = 0}^{i-1} I_{n^j} \otimes A \otimes I_{n^{i-1-j}}, \\
    \mathcal{A}_0 &:= \sum_{ j = 0}^{i-1} I_{n^j} \otimes A_0 \otimes I_{n^{i-1-j}}. \\
\end{split}
\end{equation}
then $\delta_{P}$ solves Problem \ref{prob1} and $\delta_{P} \geq \delta$.
\end{proposition}
\begin{proof}
Choose integer $i \geq 1$ and choose $\delta \geq 0$ satisfying the constraint \eqref{prop1} at the $i^{\rm th}$ level for some $P \in S_{++}^{n^i}$. Then \eqref{uncertain} is stable with respect to $\delta$ as a result of Proposition \ref{thing5}.  From this, we know that $A$ is Hurwitz; this is due to the fact that when \eqref{uncertain} is stable with respect to $\delta$, \eqref{uncertain} is stable with respect to $0$.  Moreover, the optimization problem \eqref{deltaP} is always feasible as $\delta$ satisfies the constraint \eqref{const5}.

The constraint \eqref{const5} is equivalent to the constraint \eqref{prob1}, where here $P \in S_{++}^{n^i}$ is fixed, and $\delta^*$ is a free variable.  Therefore, if there exists a $\delta^* \geq \delta$ satisfying \eqref{const5} for $P$ then the system \eqref{uncertain} is safe with respect to $\delta^*$, \emph{i.e.} $\delta^*$ solves Problem \ref{prob1},  and $\delta^* \geq \delta$. This completes the proof.
\end{proof}

We next apply the procedures detailed in Propositions \ref{thing5} and \ref{corbin2} to form an iterative algorithm for under-approximating $\hat{\delta}$ with reduced conservatism.  This procedure is given in Algorithm \ref{alg:one}.

\begin{theorem}
If $A$ is Hurwitz, then the output of Algorithm \ref{alg:one} solves Problem \ref{prob1}, for all $\varepsilon \geq 0$ and all integers $i_{\rm max} \geq 1.$
\end{theorem}

\begin{algorithm}[t]
\caption{Solve Problem \ref{prob1}}
\begin{algorithmic}[1]
\setlength\tabcolsep{0pt}
\Statex\begin{tabulary}{\linewidth}{@{}LLp{6cm}@{}}
\textbf{input}&:\:\:& $A$, $A_0$ from \eqref{uncertain}.  Increment $\varepsilon \geq 0$. Maximum level $i_{\rm max} \geq 1.$\\
\textbf{output}&:\:\:& $\underline{\delta}$ solving Problem \ref{prob1}. Lyapunov certificate $V(x)$.\\
&&
\end{tabulary}
\Function{UnderApprox}{$A,\, A_0$}
\State \textbf{initialize: } $\delta \leftarrow 0$
\While{True}
\State $\delta' \gets \delta + \varepsilon$
\For{$i\in \{1,\, \cdots,\, i_{\rm max}\}$}
\State $\S \gets \{A,\, A+\delta' A_0\}$
\State Compute $P_i \in S^{n^{i}}_{++}$ satisfying \eqref{prop1} for $\delta'$
\If{$P_i$ found}
\State $\delta \gets \delta'$
\State Break
\ElsIf{$i == i_{\rm max}$}
\State \textbf{return} $\delta$
\State \textbf{return} $V(x) = (\otimes^i x)^TP_i(\otimes^i x)$
\EndIf
\EndFor
\State Compute $\delta_{P_i}$ satisfying \eqref{deltaP} for $P_i$
\If{$\delta_{P_i}$ found}
\State $\delta \gets \delta_{P_i}$
\EndIf
\EndWhile
\EndFunction
\State\textbf{end function}
\end{algorithmic}
\label{alg:one}
\end{algorithm}

\section{Computing the Worst-Case Switching Function}

The Lyapunov function returned by Algorithm \ref{alg:one} proves asymptotic stability of \eqref{uncertain} with $\Delta(t) \in [0,\,\underline{\delta}]$. We should be able to identify $\hat{\delta}$ with a procedure such as Algorithm \ref{alg:one} since a homogeneous polynomial Lyapunov function is a necessary condition for the stabiity of \eqref{uncertain}. However, computing a Lyapunov function of sufficiently high order via semidefinite programming techniques may become computationally intractable for higher dimensional systems or as $\delta \rightarrow \hat{\delta}$. Therefore, we lay the groundwork for a numerical procedure to produce an upper bound for $\hat{\delta}$ in this section. 

The Lyapunov function returned by Algorithm \ref{alg:one} includes valuable information about how to shape $\Delta(t)$ in order to produce a periodic or diverging trajectory. Specifically, choosing a function $\Delta(t) \in [0,\; \delta]$ which maximizes the time-derivative of the Lyapunov function returned by Algorithm \ref{alg:one} will push the system as close as possible to instability for the given $\delta$, especially as the order of the Lyapunov function is increased. This intuition is formalized in Proposition \ref{deltamin}, and we now use it to approximate the worst-case switching function $\Delta_w(t;\delta)$. 

\begin{proposition}\label{deltamin}
Let \eqref{lyap} be a Lyapunov function, parameterized by $P$, for the system \eqref{uncertain}. Then the function $\Delta_w(t;\delta,P)$ which maximizes $\dot{V}(x(t))$ is given by

\begin{align} \label{deltax}
    \Delta_w(t;\delta,P) =
    \begin{cases}
     0,  \; &\text{if } \mathcal{I}(t) < 0 \\
     \delta, \; & \text{otherwise}
    \end{cases}
\end{align}

where the indicator function $\mathcal{I}(t)$ is given by

\begin{align}\label{indicator}
\mathcal{I}(t) = (\otimes^i x(t))^T(\A_{0}^TP + P\A_0)(\otimes^i x(t))
\end{align}
for $\A_0=\A_0^i$.
\end{proposition}

\begin{proof}
The optimization problem
\begin{equation}\label{optproblem}
    \begin{array}{rcl}
    \Delta_w(t;\delta,P) = \argmax\limits_{\Delta} & \dot{V}(x(t)) \vspace{.1in}\\
     \mbox{such that} & \Delta \geq 0 \\
     & \Delta \leq \delta
    \end{array}
\end{equation}
with 
\begin{multline*}\small
\dot{V}(x(t)) = (\otimes^i x(t))^T((\A+\Delta \A_0)^TP \\
+ P(\A+\Delta \A_0))(\otimes^i x(t))
\end{multline*}
for $\A = \A^i$ and $\A_0=\A_0^i$ is solved by \eqref{deltax}.
\end{proof}

The optimization problem \eqref{optproblem} shows that the $\Delta(t)$ selected is the value in the interval $[0,\delta]$ which maximizes $\dot{V}(x(t))$ for all $t \geq 0$. The state-dependence of $\mathcal{I}(t)$ suggests that $\Delta_w(t;\delta,P)$ is not unique and generally differs for trajectories beginning at different initial conditions $x_0$. Therefore, Algorithm \ref{alg:two} uses a simulation given an initial condition $x_0$ and a parameter $\delta$ and numerically solves \eqref{indicator} in order to produce \eqref{deltax}. Moreover, the system \eqref{uncertain} need not be stable with respect to $\delta$ in order to produce \eqref{deltax}.

\begin{remark}
Since $\Delta_w(t;\delta,P)$ is a piecewise-constant function, it is most easily utilized in a numerical procedure when expressed as a pair of finite sets $T = \{t_0,\dots, t_f\}$ and $\Sigma = \{\sigma_0, \dots, \sigma_{f-1}\}$ such that

\begin{equation}\label{deltapwc}
    \Delta_w(t;\delta,P) =
    \begin{cases}
     \sigma_0,  \; &\text{for } t_0 \leq t < t_1 \\ & \dots \\
     \sigma_{f-1}, \; & \text{for } t_{f-1} \leq t < t_f
    \end{cases}
\end{equation}
\end{remark}

\begin{algorithm}[t]
\caption{Solve Problem 2}
\begin{algorithmic}[1]
\setlength\tabcolsep{0pt}
\Statex\begin{tabulary}{\linewidth}{@{}LLp{6cm}@{}}
\textbf{input}&:\:\:&  $A,\,A_0$ from \eqref{uncertain}, desired margin $\delta$.\\
&& $P$ from Algorithm \ref{alg:one}. \\
&& Simulation parameters: initial condition $x_0$ and time horizon $t_f$.\\
\textbf{output}&:\:\:& $T,\,\Sigma$ from \eqref{deltapwc} to describe $\Delta_w(t;\delta,P)$.
\end{tabulary}
\Statex
\Function{FindSwitchingSequence}{inputs}
\State \textbf{initialize: }
\par
\hskip\algorithmicindent$\A^i,\,\A_0^i \leftarrow$ from \eqref{hierarchy}
\par
\hskip\algorithmicindent$t_0\leftarrow 0$
\par
\hskip\algorithmicindent$\sigma_0 \leftarrow$ from \eqref{deltax}
\par
\hskip\algorithmicindent$T \leftarrow \{t_0\}$
 \par
\hskip\algorithmicindent$\Sigma \leftarrow \{\sigma_0\}$ 
\par
\hskip\algorithmicindent$k \leftarrow 1$
\While{$t_{k-1} < t_f$}
\State Numerically solve:
\State\hskip\algorithmicindent $t_{k} \leftarrow t_{k-1} +  \argmin\limits_{t > 0} \quad \mathcal{I}(t - t_{k-1})=0$
\State $T\leftarrow\text{append}(t_k)$
\If{$\Sigma(k-1) == 0$}
\State $\Sigma(k) \leftarrow \delta$
\Else
\State $\Sigma(k) \leftarrow 0$
\EndIf
\State $k = k + 1$
\EndWhile
\State \textbf{return} $T,\,\Sigma$
\EndFunction
\State \textbf{end function}
\end{algorithmic}
\label{alg:two}
\end{algorithm}

The worst-case switching function \eqref{deltax} contains all of the necessary information needed to compute an upper bound on the stability margin of \eqref{uncertain}.

\section{Bounding the Stability Margin From Above and Producing a Periodic Trajectory}

We now seek to approximate the lowest upper bound on the stability margin of \eqref{uncertain} with an iterative procedure which, starting with $\delta = \underline{\delta}$, increments $\delta$ until $A_w(t;\delta,P)$ can be shown to produce a periodic trajectory. Such a trajectory with an initial condition $x_0$ and time horizon $t_f$ can become periodic, but it is not necessarily periodic from $x_0$. Therefore, the switching sequence described by $T$ and $\Sigma$ in \eqref{deltapwc} is searched to see if there exist indices $k \geq j \geq 0$ such that a discrete transition matrix $A_d$ exists with an eigenvalue with magnitude 1.

\begin{proposition}\label{prop:stopcrit}
A switching function described by \eqref{deltapwc} with sets $T$ and $\Sigma$ produces a limit cycle for \eqref{uncertain} if there exists some $k \geq j \geq 0$ such that the discrete transition matrix 
\begin{multline}\label{eq_Ad}
    A_d = \\
    e^{(A+\Sigma(k-1)A_0)(t_k - t_{k-1})} \cdots e^{(A+\Sigma(j)A_0)(t_{j+1} - t_j)}
\end{multline}
has an eigenvalue of magnitude equal to 1.
\end{proposition}

Algorithm \ref{alg:three} seeks to find the parameter $\overline{\delta}$ which induces a limit cycle. In doing so, it calls Algorithm \ref{alg:two} to compute $\Delta_w(t;\delta,P)$ and searches the switching sequence for a discrete transition matrix $A_d$ that has an eigenvalue of magnitude 1. Such an $A_d$ certifies that $\Delta_w(t;\overline{\delta},P)$ produces a periodic trajectory.

\begin{remark}
The $\overline{\delta}$ found by Algorithm \ref{alg:three} depends on the parameter $P$ returned by Algorithm \ref{alg:one}. A higher-order Lyapunov function will cause a smaller value of $\overline{\delta}$ to be found. Therefore, $\overline{\delta}$ is an upper bound on $\hat{\delta}$, and in order to find the lowest possible $\overline{\delta}$, a high-order Lyapunov function should be used with Algorithm \ref{alg:three}, if possible.
\end{remark}

\begin{algorithm}[t]
\caption{Solve Problem 3}
\begin{algorithmic}[1]
\setlength\tabcolsep{0pt}
\Statex\begin{tabulary}{\linewidth}{@{}LLp{6cm}@{}}
\textbf{input}&:\:\:&  $A,\,A_0$ from \eqref{uncertain}, $\underline{\delta}$ from Alg. \ref{alg:one}.\\
&& $P$ from Algorithm \ref{alg:one}. \\
&& Simulation parameters: initial condition $x_0$ and time horizon $t_f$.\\
\textbf{output}&:\:\:& $\overline{\delta}$, an upper bound on $\hat{\delta}$. \\ && $T$, $\Sigma$ from \eqref{deltapwc} which produces a periodic trajectory.\\
&& $A_d$, discrete transition matrix certifying marginal stability.
\end{tabulary}
\Statex
\Function{FindPeriodicTrajectory}{inputs}
\State  \textbf{initialize:}
\hskip\algorithmicindent\par $\delta \leftarrow \underline{\delta}$
\While{No Periodic Trajectory Found}
\If{$A + \delta A_0$ not Hurwitz}
\State \% Trivial Switching Sequence Found
\State $A_d = e^{(A + \delta A_0)(t_f)}$
\State \textbf{return} $\overline{\delta} = \delta$, $A_d$, $\Sigma = \{\delta\}$, $T = \{0,t_f\}$
\EndIf
\State $T,\,\Sigma =$ \Call{FindSwitchingSequence}{inputs}
\State Search $T,\,\Sigma$ for $k \geq j \geq 0$ that satisfies \eqref{eq_Ad}.
\If{$A_d$ from \eqref{eq_Ad} found}
\State \textbf{return} $\overline{\delta}=\delta$, $A_d$, $T = T(j:k)- T(j)$, \par\hskip \algorithmicindent\hskip\algorithmicindent\hskip\algorithmicindent\hskip\algorithmicindent$\Sigma = \Sigma(j:k-1)$
\EndIf
\State $\delta \leftarrow $ increment
\EndWhile
\EndFunction
\State \textbf{end function}
\end{algorithmic}
\label{alg:three}
\end{algorithm}

\section{Numerical Examples and Applications}
The practical importance of the techniques presented in this article is highlighted in the following examples.
\begin{example}\label{ex1}[Second Order System] We first study the system \eqref{uncertain} with $n=2$ and

\begin{gather} \label{2dex}
\begin{split}
A= 
\begin{bmatrix} 
0 & 1 \\ -1 & -0.5
\end{bmatrix}, \;
A_0 =
\begin{bmatrix}
0 & 0 \\ -1 & 0
\end{bmatrix}.
\end{split}
\end{gather}

This system resembles a spring-mass damper with uncertainty in the spring constant. Algorithm \ref{alg:one} is used to compute a lower bound on the stability margin. The algorithm returns  $\underline{\delta} = 2.15$ with a 14\textsuperscript{th} order polynomial Lyapunov function. By increasing $i_{max}$ to compute a 28\textsuperscript{th} order Lyapunov function, the algorithm computes $\underline{\delta} = 2.16$. Beyond that, only slight increases in $\underline{\delta}$ can be obtained for significantly higher-order Lyapunov functions such that computation quickly becomes intractable. The next step is to find a $\overline{\delta}$ for which we can produce a periodic trajectory.

Using the 14\textsuperscript{th} order Lyapunov function, Algorithm \ref{alg:three} provides an upper bound $\overline{\delta}$ on the system's stability margin and provides the subset of the switching sequence from Algorithm \ref{alg:two} which certifies periodicity of the system (or near periodicity---the numerical method is approximate) with $A_w(t;\overline{\delta},P)$. The value $\overline{\delta}=2.21$ was computed, and two switching sequences were returned: one using simulation data from $x_0 = [1 \; 1]^T$ and another with $x_0 = [-0.2 \; 0.8]^T$, both with a time horizon $t_f=20$. The difference $\overline{\delta} - \underline{\delta} = 0.06$ illustrates how the approximation of a stability margin using Algorithm \ref{alg:one} is likely conservative.

\begin{center}
\begin{tabular}{ll}
    \multicolumn{2}{l}{Switching Sequence 1 using $x_0 = [1 \; 1]^T$} \\
    \hline
    $T_1$ & \{0, 0.943, 2.374, 
    3.317, 4.747\} \\
    $\Sigma_1$ & \{2.21, 0, 2.21, 0\}\vspace{.5em} \\
    \multicolumn{2}{l}{Switching Sequence 2 using $x_0 = [-0.2 \; 0.8]^T$} \\
    \hline
    $T_2$ & \{0, 0.251, 1.681, 2.624, 4.055\} \\
    $\Sigma_2$ & \{2.21, 0, 2.21, 0\}
\end{tabular}
\end{center}

Figure \ref{fig:Example_2d} shows two trajectories from $x_0=[-0.2\;0.8]^T$, one using $A_w(t;\overline{\delta},P_{14})$ described by $(T_1,\Sigma_1)$ and the other using $A_w(t;\overline{\delta},P_{14})$ described by $(T_2,\Sigma_2)$. Both switching sequences produce a limit cycle but at different locations. This demonstrates that a function $\Delta_w(t;\delta,P)$ is not unique but that any such function can produce a periodic trajectory for the same value of $\delta$.

\begin{figure}
		\centering
		\includegraphics[width=.5\textwidth,angle = 0]{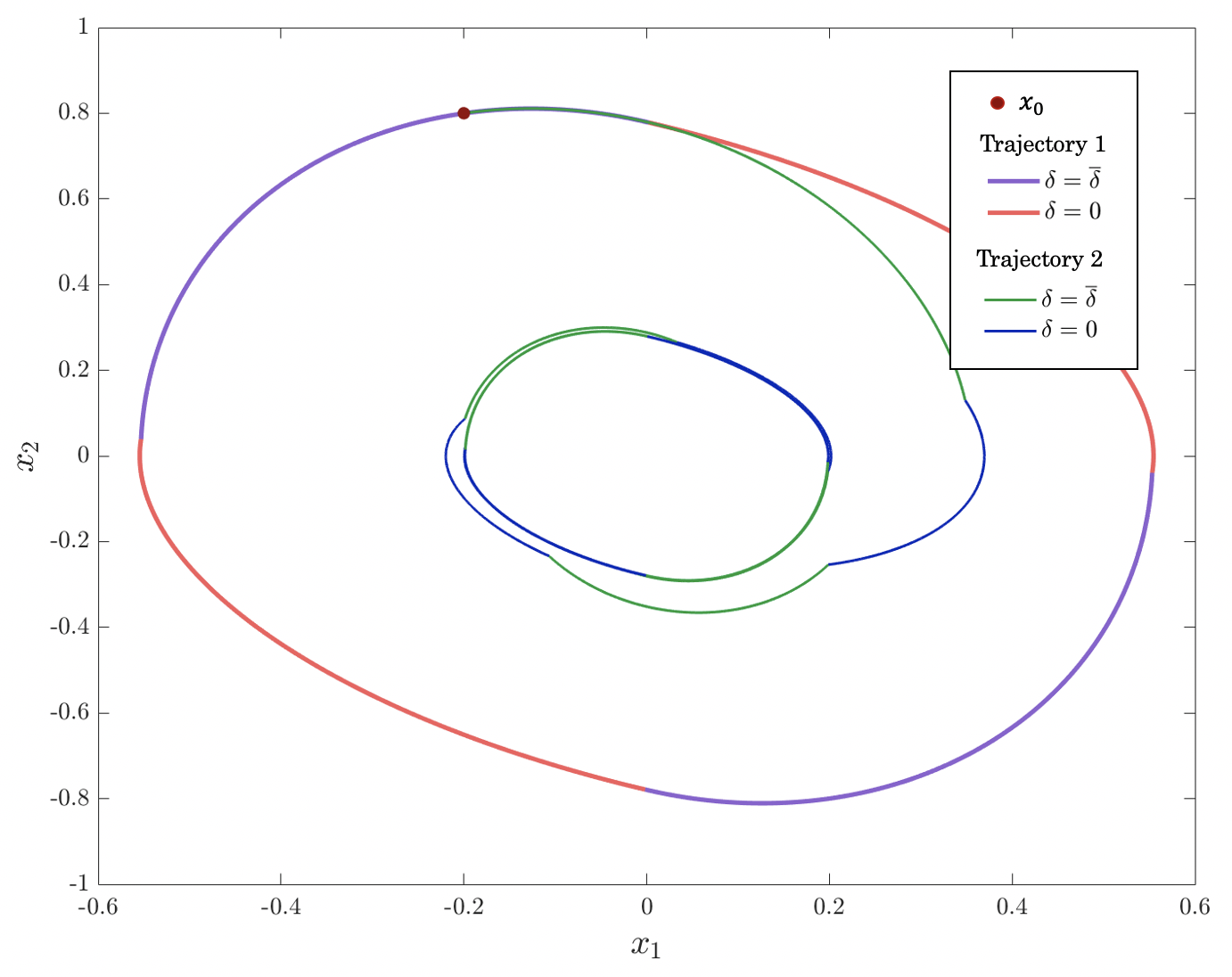}\\
		\caption{Example \ref{ex1}. Plot of two trajectories beginning at $x_0$ but with two different switching functions $\Delta_w(t;\overline{\delta}=2.21,P_{14})$. Both result in periodic trajectories.}
		\label{fig:Example_2d}
\end{figure}

\end{example}

\begin{example}\label{ex2}[Fourth Order System] For higher-order systems, limit cycles cannot be visualized on the phase plane; instead we search for periodic orbits via analysis of the indicator function \eqref{indicator} and Algorithm \ref{alg:three}. To demonstrate this, we study \eqref{uncertain} with 4\textsuperscript{th} order dynamics inspired by the linearized, non-dimensional lateral dynamics of a fixed-wing aircraft with the following parameters \cite{etkin}:

\begin{gather}
\begin{split}
A=&
\begin{bmatrix} 
-3.088 & 0 & -1425.042 & 4.5956 \\ -18.906 & -166.878 & 29.223 & 0 \\
6.762 & 4.445 & -19.389 & 0 \\ 0 & 1428.6 & 0 & 0
\end{bmatrix}, \\
A_0 =&
\begin{bmatrix}
-1 & 0 & -10 & 10 \\ -10 & -10 & 10 & 0 \\ 10 & 10 & -10 & 0 \\ 0 & 10 & 0 & 0
\end{bmatrix}.
\end{split}
\end{gather}

Algorithm \ref{alg:one} computes a maximum value of $\underline{\delta}=0.24$ for a 6th-order polynomial Lyapunov function. It may be possible to find a larger $\underline{\delta}$, but increasing $i_{max}$ to attain a higher-order Lyapunov function would make computation intractable without  additionally employing an algorithm to reduce the redundancies produced by the use of the Kronecker product.

To get a better picture of system robustness, we run Algorithm \ref{alg:three} using the 6th-order Lyapunov function, and compute $\overline{\delta} = 0.27$. The switching sequence $\Delta(t;\overline{\delta},P_{6})$ which produces a periodic trajectory is returned, described by $(T,\Sigma)$:

\begin{center}
\begin{tabular}{ll}
    \multicolumn{2}{l}{Switching Sequence using $x_0 = [1 \; 1\; 1 \; 1]^T$} \\
    \hline
    $T_1$ & \{0, 0.027, 0.060\} \\
    $\Sigma_1$ & \{0.27, 0\}\vspace{.5em} \\
\end{tabular}
\end{center}

Figure \ref{fig:indicator} plots the indicator function \eqref{indicator} for the cases where $\Delta_w(t;\overline{\delta},P_6)$ and $\Delta_w(t;\underline{\delta},P_6)$ are used. The former generates a limit cycle starting from about $t=0.332$, while the latter produces values of $\mathcal{I}(t)$ which slowly decay over time. The initial erratic jumps in $\mathcal{I}(t)$ show that the trajectory using $\overline{\delta}$ does not begin on a limit cycle but rather takes a few tenths of a second to reach the periodic orbit. \qedhere

\begin{figure}
		\centering
		\includegraphics[width=.5\textwidth,angle = 0]{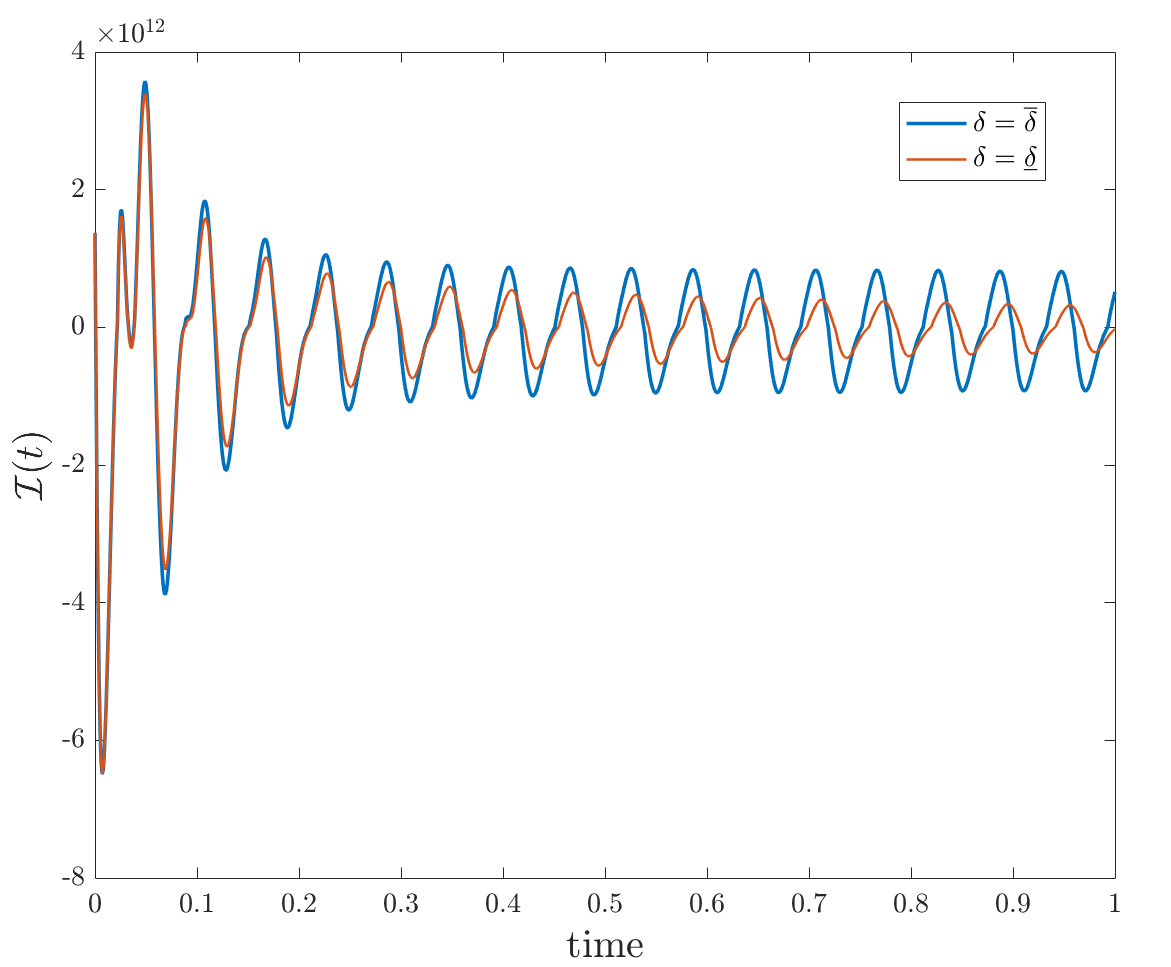}\\
		\caption{Example \ref{ex2}. Plot of $\mathcal{I}(t)$ produced with $\Delta_w(t;\delta,P_6)$ for $\delta=\overline{\delta}$ and $\delta=\underline{\delta}$.}
		\label{fig:indicator}
\end{figure}

\end{example}

Of practical interest to an analysis of uncertain systems are worst-case impulse responses. In \cite{abate2020performance}, the authors use the system hierarchy \eqref{hierarchy} to easily compute exponential bounds on an impulse response. The methods in this paper can be used to find a worst-case impulse response which approaches these bounds. We add a control input and define an output for the system \eqref{uncertain} and study
\begin{align} \label{ex_system}
\begin{split}
    \dot{x} &= A(t)x +Bu\\
    y &= Cx \\
    A(t) &= A + \Delta(t) A_0
\end{split}
\end{align}

The impulse response $h(t)$ of \eqref{ex_system} is described parametrically by a solution $\varphi (t)$ to 

\begin{equation}\label{impulse_sys}
\begin{split}
    \dot{\varphi}(t) & = A(t)\varphi(t) \\
    h(t) &= C\varphi(t) \\
    \varphi(0) &= B
\end{split}
\end{equation}

Since the impulse response of \eqref{ex_system} is closely related to the unforced response, the algorithms presented above can be used to compute various worst-case and marginally stable trajectories in these scenarios as well. 

\begin{example}[Worst-Case Impulse Response] \label{ex3}
Consider the system studied in Example \ref{ex1} with the additional parameters $B=[0 \; 1]^T$ and $C=[1 \; 0]$. We use Theorem 4 in \cite{abate2020performance}, with parameter $\alpha=0.11$, to produce an exponential bound on the impulse response. Using Algorithm \ref{alg:two}, a worst-case switching function $\Delta_w(t;\delta=1,P_{14})$ is computed for \eqref{impulse_sys}. Figure \ref{fig:impulse} shows how this switching function produces an impulse response which approaches the bound much more closely than the system with no switching.

\begin{figure}
		\centering
		\includegraphics[width=.5\textwidth,angle = 0]{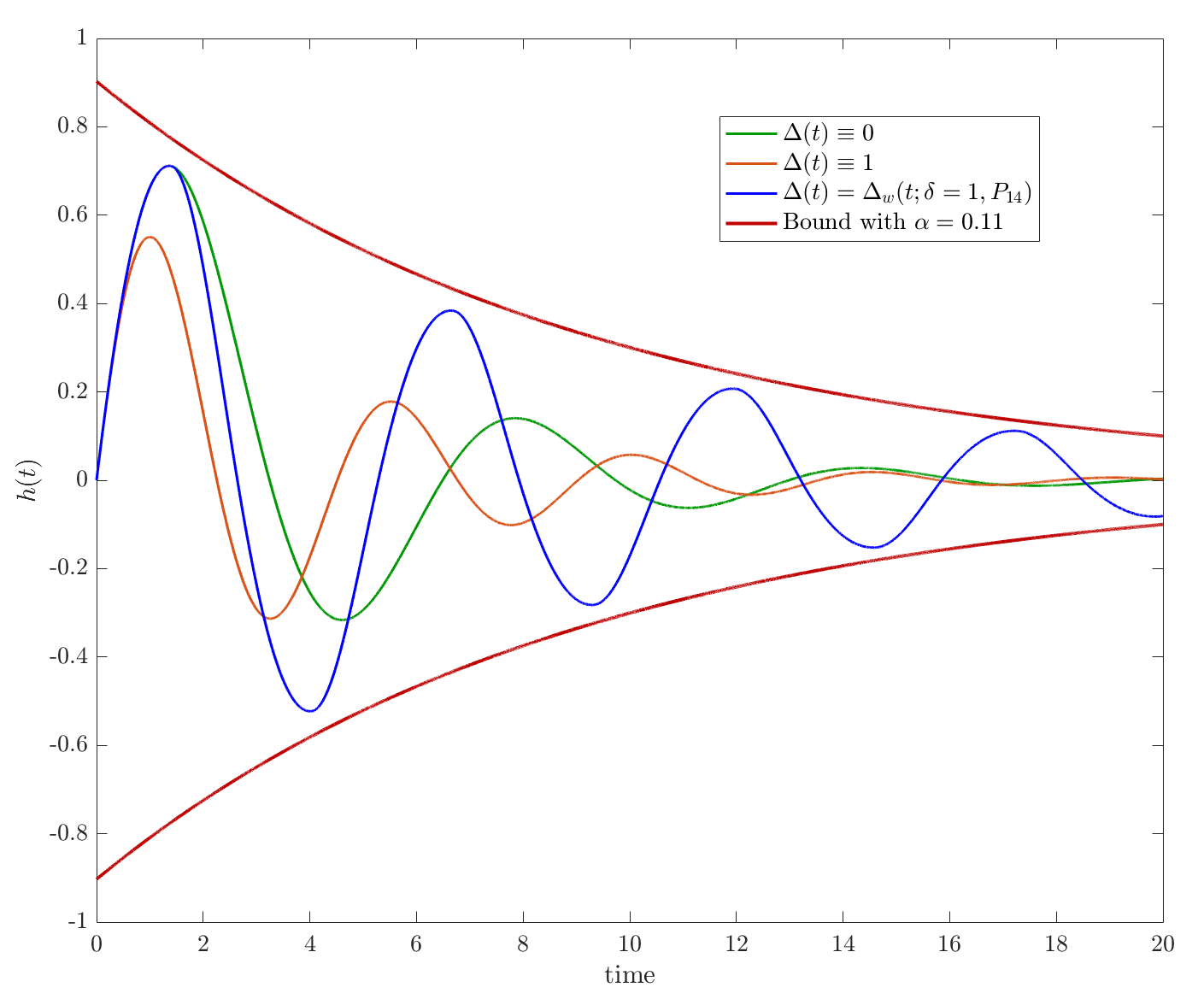}\\
		\caption{Example \ref{ex3}. Impulse response plots showing how using the worst-case switching function can cause an impulse response to approach an exponential bound.}
		\label{fig:impulse}
\end{figure}
\end{example}

\section{Conclusion}
In this work, we show how the more general class of homogeneous polynomial Lyapunov functions is used to compute stability margins with reduced conservatism, and we show how these Lyapunov functions aid in the search for periodic trajectories for marginally stable LTV systems.  Our work is premised on the recent observation that the search for a homogeneous polynomial Lyapunov function for some LTV systems is easily encoded as the search for a quadratic Lyapunov function for a related LTV system, and our main contribution is an intuitive algorithm for generating upper and lower bounds on the system's stability margin. We show also how the worst-case switching scheme---which draws an LTV system closest to a periodic orbit---is generated.


\bibliography{Bibliography}

\begin{thebibliography}{10}

\bibitem{liberzon2003switching}
D.~Liberzon, {\em Switching in systems and control}.
\newblock Springer Science \& Business Media, 2003.

\bibitem{afman2018triple}
J.-P. Afman, E.~Feron, and J.~Hauser, ``Triple-integral control for reduced-g
  atmospheric flight,'' in {\em 2018 Annual American Control Conference (ACC)},
  pp.~392--397, IEEE, 2018.

\bibitem{young1997lower}
P.~M. Young and J.~C. Doyle, ``A lower bound for the mixed/spl mu/problem,''
  {\em IEEE Transactions on Automatic Control}, vol.~42, no.~1, pp.~123--128,
  1997.

\bibitem{feronBoussios}
C.~Boussios and E.~Feron, ``Estimating the conservatism of {P}opov's criterion
  for real parametric uncertainties,'' {\em Systems \& control letters},
  vol.~31, no.~3, pp.~173--183, 1997.

\bibitem{BEFB:94}
S.~Boyd, L.~{El~{G}haoui}, E.~Feron, and V.~Balakrishnan, {\em Linear Matrix
  Inequalities in System and Control Theory}, vol.~15 of {\em SIAM Studies in
  Applied Mathematics}.
\newblock SIAM, 1994.

\bibitem{Rap:90}
L.~B. Rapoport, ``Boundary of absolute stability of nonlinear nonstationary
  systems and its connection with the construction of invariant functions,''
  {\em Automation and Remote Control}, vol.~51, no.~10, pp.~1368--1375, 1990.

\bibitem{Rap:93}
L.~B. Rapoport, ``On existence of non-smooth invariant functions on the
  boundary of the absolute stability domain,'' {\em Automation and Remote
  Control}, vol.~54, no.~3, pp.~448--453, 1993.

\bibitem{parrilo_jsr}
P.~A. Parrilo and A.~Jadbabaie, ``Approximation of the joint spectral radius
  using sum of squares,'' {\em Linear Algebra and its Applications}, vol.~428,
  no.~10, pp.~2385--2402, 2008.

\bibitem{abate2019lyapunov}
M.~Abate, S.~Coogan, and E.~Feron, ``Lyapunov differential equation hierarchy
  and polynomial {L}yapunov functions for switched linear systems,'' {\em arXiv
  preprint arXiv:1906.04810}, 2019.

\bibitem{parrilo2000structured}
P.~A. Parrilo, {\em Structured semidefinite programs and semialgebraic geometry
  methods in robustness and optimization}.
\newblock PhD thesis, California Institute of Technology, 2000.

\bibitem{abate2020performance}
M.~Abate, C.~Klett, S.~Coogan, and E.~Feron, ``Performance analysis and
  non-quadratic lyapunov functions for linear time-varying systems,'' {\em
  arXiv preprint arXiv:2009.00727}, 2020.

\bibitem{miller2020peak}
J.~Miller, D.~Henrion, and M.~Sznaier, ``Peak estimation and recovery with
  occupation measures,'' {\em arXiv preprint arXiv:2009.06120}, 2020.

\bibitem{ahmadi2017sum}
A.~A. Ahmadi and P.~A. Parrilo, ``Sum of squares certificates for stability of
  planar, homogeneous, and switched systems,'' {\em IEEE Transactions on
  Automatic Control}, vol.~62, no.~10, pp.~5269--5274, 2017.

\bibitem{mason2006common}
P.~Mason, U.~Boscain, and Y.~Chitour, ``Common polynomial {L}yapunov functions
  for linear switched systems,'' {\em SIAM journal on control and
  optimization}, vol.~45, no.~1, pp.~226--245, 2006.

\bibitem{etkin}
B.~Etkin, {\em Dynamics of Atmospheric Flight}.
\newblock Courier Corporation, 2005.

\end{thebibliography}
\bibliographystyle{ieeetr}
\end{document}